\documentclass[12pt, reqno]{amsart}

\usepackage{amssymb}
\usepackage{float}
\usepackage{graphicx}
\usepackage{tabularx}
\usepackage{multirow}
\usepackage{geometry}
\usepackage{subfigure}
\usepackage{caption}
\usepackage{xcolor}
\theoremstyle{plain}

\newtheorem{corollary}{Corollary}[subsection]

\newtheorem{lemma}{Lemma}

\newtheorem{proposition}{Proposition}[subsection]
\newtheorem{remark}{Remark}[subsection]

\newtheorem{theorem}{Theorem}[subsection]

\numberwithin{equation}{section}

\usepackage{fancyhdr}
\usepackage{hyperref}

\begin{document}
\title{A set of nonlinear coherent states for the pseudoharmonic oscillator}
\begin{footnotesize}
\email{$^{\star}$ahbli.khalid@gmail.com \\  $^{\flat}$hamidoukass@gmail.com\\
$^\ddagger$Kayupepatrick@gmail.com \\ $^{\Upsilon}$akouraich@yahoo.fr }
\end{footnotesize}
\maketitle
\begin{center}
\author{K. Ahbli$^{\star}$, H. Kassogue$^{\flat}$, P.K. Kayupe$^\ddagger$ \ and
\ A. Kouraich$^{\Upsilon}$\bigskip }
\end{center}
\begin{center}
\begin{scriptsize}
$^{\star}$ Faculty of Sciences of Agadir, Ibn Zohr University, Morocco\\
$^{\flat}$ FST of Tangier, Abdelmalek Essaadi University, Morocco \\
$^\ddagger$  13, Market avenue, Mongt-Ngafula, Kinshasa, DR Congo \\
$^{\Upsilon}$ FST of B\'eni Mellal, Sultan Moulay Slimane University, Morocco \vspace*{0.2mm}\vspace*{0.2mm}
\end{scriptsize}
\end{center}

\begin{abstract} 
	We construct two-parameters family of nonlinear coherent states by replacing the factorial in coefficients $z^n/\sqrt{n!}$ of the canonical coherent states by a specific generalized factorial $x_n^{\gamma,\sigma}!$ where parameters $\gamma$ and $\sigma$ satisfy some conditions for which the normalization condition and the resolution of identity are verified. The obtained family is a generalization of the Barut-Girardello coherent states and those of the philophase states. In the particular case of parameters $\gamma$ and $\sigma$, we attache these states to the pseudo-harmonic oscillator depending on two parameters $\alpha,\beta> 0$. The obtained nonlinear coherent states are superposition of eigenstates of this oscillator. The associated Bargmann-type transform is defined and we derive some results.
\end{abstract}

\section{Introduction}
Coherent states (CSs) were first discovered by Schr\"{o}dinger in 1926 \cite{Schrodinger} as wavepackets having dynamics similar to that of a classical particle submitted to a quadratic potential. They have arised from the study of the quantum harmonic oscillator to become very useful in different areas of physics. Nonlinear coherent states (NLCSs), which can be classified as an algebraic generalization of the canonical CSs for the harmonic oscillator, were implicitly defined by Shanta \textit{et al} \cite{Shanta} in a compact form and introduced explicitly by de Matos Filho and Vogel \cite{Matos} and Man'ko \textit{et al} \cite{Manko}. This notion attracted much attention in recent decades, especially because of their nonclassical properties in quantum optics \cite{RT}.

In this paper, we construct a two-parameters family of NLCSs, denoted $|z,\gamma ,\sigma \rangle $, by replacing the factorial $n!$ in coefficients $%
z^{n}/\sqrt{n!}$ of the canonical CSs by a specific generalized
factorial $x_{n}^{\gamma,\sigma }!:=x_{1}^{\gamma,\sigma }\cdots x_{n}^{\gamma,\sigma }$ with $x_{0}^{\gamma,\sigma }=0$, where $x_{n}^{\gamma,\sigma}$ is a sequence of positive numbers (given by \eqref{3.1} see below) and $(\gamma,\sigma)\in S_1\cup S_2$ where $S_1$ and $S_2$ are some subsets of $\mathbb{R}^2$ obtained from restrictions imposed by normalization conditions and the resolution of identity relation which must be satisfied by these states in an arbitrary Hilbert space $\mathcal{H}$ of quantum states (often termed Fock space). The obtained NLCSs are a generalization of the Barut Girardello CSs type \cite{Barut} obtained when $2\gamma =1,2,3,...$ and $\sigma=0$ and those of the philophase states \cite{Brif} occuring when $\gamma=1/2$ and $\sigma$ being a positive integer. The case $\gamma=1/2$ and $\sigma=0$, have been considered in \cite{KKM} where the authors associated to their NLCSs two set of orthogonal polynomials by following the work of T. Ali and M. Ismail \cite{AI}. Next, we consider the model of the pseudoharmonic oscillator (PHO) $\Delta_{\alpha ,\beta}$ \cite{Godman}  acting on the Hilbert space of square integrable functions on the positive real half-line $L^{2}(\mathbb{R}_{+},d\xi)$. Then, we construct the NLCSs attached to $\Delta_{\alpha ,\beta}$ by choosing $\mathcal{H}= L^{2}(\mathbb{R}_{+},d\xi)$ and the basis vectors as the eigenfunctions of this oscillator which constitutes a complete orthonormal basis of $L^{2}(\mathbb{R}_{+},d\xi)$. The wavefunctions of constructed NLCSs are obtained in the special case $\gamma =2^{-1}\mu(\alpha)$. Recently, B. Mojaveri \textit{et al} \cite{Mojaveri2018} have introduced and studied three kinds of NLCSs associated with the para-Bose oscillator which is a particular case of the PHO obtained (up to scale factor$=\frac{1}{2}$) for $\alpha=(p-1)(p-3)/4$ and $\beta =1$. Finally, we exploit the obtained result to define a new Bargmann-type integral transform and we derive some interesting results.

The rest of the paper is organized as follows. In Section 2, we summarize the construction of NLCSs. In Section 3, we particularize the formalism of NLCSs for the sequences $x_n^{\gamma,\sigma}$ and we discuss the corresponding resolution of the identity. The  $x_n^{\gamma,\sigma}$-NLCSs and Bargmann-type transform attached to PHO are defined in Section 4. Section 5 is devoted to the conclusion.

\section{ Nonlinear coherent states formalism}
This section devoted to a quick review on the construction of NLCSs. Details and proofs of statements may be found in \cite[pp.146-151]{AAG}. The principal idea is to involve a new sequence of positive numbers in the superposition coefficients. More precisely, let us first recall the
series expansion definition of the canonical CSs, which first
was due to Iwata \cite{Iwata}:
\begin{equation}
|z\rangle =(e^{z\bar{z}})^{-1/2}\sum\limits_{n=0}^{\infty }\frac{\bar{z}^{n}}{\sqrt{n!}}|\psi_{n}\rangle ,\quad z\in \mathbb{C}.  \label{2.1}
\end{equation}
The kets $|\psi_{n}\rangle ,\,\,n=0,1,2,...,\infty$ are an
orthonormal basis in an arbitrary (complex, separable, infinite dimensional)
Hilbert space $\mathcal{H}$. The related NLCSs are constructed as follows.
 
 Let $\{x_{n}\}_{n=0}^{\infty }$, be an infinite sequence of positive
	numbers with $\lim_{n\rightarrow +\infty }x_{n}=R^{2}$ where $R>0$ could be
	finite or infinite, but not zero. We define the generalized factorial by $%
	x_{n}!=x_{1}x_{2}\cdots x_{n}$ and $x_{0}!=1$. For each $z\in \mathcal{D}$ a complex domain, the NLCSs constituting a generalization of \eqref{2.1} are defined by
	\begin{equation}
	|z\rangle =(\mathcal{N}(z\bar{z}))^{-1/2}\sum\limits_{n=0}^{\infty }\frac{%
		\bar{z}^{n}}{\sqrt{x_{n}!}}|\psi _{n}\rangle ,\quad z\in \mathcal{D}
	\label{2.2}
	\end{equation}
	where the normalization factor 
	\begin{equation}
	\mathcal{N}(z\bar{z})=\sum\limits_{n=0}^{\infty }\frac{|z|^{2n}}{x_{n}!}
	\label{2.3}
	\end{equation}
	is chosen so that the vectors \eqref{2.2} are
	normalized to one and are well
	defined for all $z$ for which the sum \eqref{2.3} converges, i.e. $%
	\mathcal{D}=\{z\in \mathbb{C},|z|<R\}$.
We assume that there
exists a measure $d\nu $ on $\mathcal{D}$ ensuring the following resolution of the
identity
\begin{equation}
\int_{\mathcal{D}}|z\rangle \langle z|d\nu (z,\bar{z})=1_{\mathcal{H}}
\label{2.4}.
\end{equation}
\\
Setting $d\nu (z,\bar{z})=\mathcal{N}(z\bar{z})d\eta (z,\bar{z})$, it
is easily seen that in order for \eqref{2.4} to be satisfied, the measure $%
d\eta $ should be of the form $d\eta (z,\bar{z})=(1/2\pi)d\theta d\lambda (\rho ),\, z=\rho
e^{i\theta } $ where the measure $d\lambda $ solves the moment problem
\begin{equation}
\int_{0}^{R}\rho ^{2n}d\lambda (\rho )=x_{n}!,\quad n=0,1,2,...\, .
\label{2.6}
\end{equation}
In most of the practical situations, the support of the measure $d\eta $ is the whole domain $\mathcal{D%
}$, i.e., $d\lambda $ is supported on the entire interval $[0,R)$.\\
To illustrate this formalism, we consider the sequence of positive numbers
\begin{equation}
x^{\gamma}_{n}=n\left( 2\gamma +n-1\right) ,\ \ \ n=0,1,2,3,... \, ,  \label{2.7}
\end{equation}
with $2\gamma =1,2,3,...$, being a fixed parameter. Here $R=\infty$ and the moment problem is
\begin{equation}
\int_{0}^{\infty }\rho ^{2n}d\lambda (\rho )=n!(2\gamma )_{n}
\end{equation}
where $(a)_{n}=a(a+1)\cdots(a+n-1)$ with $(a)_{0}=1$, is the shifted factorial. 
The solution of this problem is
\begin{equation}
d\lambda (\rho )=\frac{2}{\pi }K_{2\gamma -1}(2\rho )\rho ^{2-2\gamma }d\rho
,\quad 0\leq \rho <\infty ,
\end{equation}
where
\begin{equation}
K_{\sigma}(x)=\frac{1}{2}\left( \frac{x}{2}\right) ^{\sigma }\int_{0}^{\infty
}\exp \left( -t-\frac{x^{2}}{4t}\right) \frac{dt}{t^{\sigma +1}},\ \ \Re (x)>0, \label{Mcd}
\end{equation}
is the Macdonald function of order $\sigma$ \cite[p.183]{Watson}. The associated coherent states are of Barut-Girardello type \cite{Barut}:
\begin{equation}
\label{BG}
|z,\gamma \rangle =\frac{|z|^{\frac{2\gamma -1}{2}}}{\sqrt{I_{2\gamma -1}(2|z|)}}%
\sum\limits_{n=0}^{\infty }\frac{\bar{z}^{n}}{\sqrt{n!(2\gamma )_{n}}}%
|\psi _{n}\rangle ,\quad z\in \mathbb{C},
\end{equation}
where $I_{\sigma}(.)$ being the modified Bessel function of the first kind and of
order $\sigma $ \cite[p.172]{Watson}.

\section{A set of NLCSs with the sequences $x_n^{\gamma,\sigma}$}
In this section, we define a new set of NLCSs attached to the sequences $x_n^{\gamma,\sigma}$, without specifying the Hilbert space $\mathcal{H}$ and the basic vectors $\left|\psi_n\right\rangle$, for which we will discuss some general properties. 
\subsection{NLCSs attached to a sequence $x_n^{\gamma,\sigma}$}\ \\
Here, we will be dealing with two-parameters family of NLCSs on the complexe plane, which generalizes the set of CSs of Barut-Girardello  \cite{Barut} type and those of the philophase states \cite{Brif} without specifying the Hamiltonian. Precisely, let us consider $\gamma>0$ and $\sigma\in\mathbb{R}\backslash\mathbb{Z}^{\ast}_{-}$, two fixed parameters and let us define the infinite sequence of positive numbers:
\begin{equation}
\label{3.1}
\quad x_0^{\gamma,\sigma}=0, \quad \text{and}\quad  x_n^{\gamma,\sigma}=\frac{(n+\sigma)^2(n+2\gamma-1)}{n},\ \ n=1,2,3,...
\end{equation}
The corresponding factorial of \eqref{3.1} reads as
\begin{equation}
\label{1.2}
x^{\gamma,\sigma}_n!=\frac{(\sigma+1)_n^2(2\gamma)_n}{n!},\ \ n=0,1,2,3,...,
\end{equation} 
where $(a)_n=a(a+1)\cdots (a+n-1)$ is the Pochhamer symbol.\\
Now, we define a set of $x_n^{\gamma,\sigma}$-NLCSs through the sequence $x_n^{\gamma,\sigma}$, under some conditions on the parameters $\gamma$ and $\sigma$, via the superposition 
\begin{equation}
\label{1.3}
\left|z,\gamma,\sigma\right\rangle=\left(\mathcal{N}_{\gamma,\sigma}\left(z\bar{z}\right)\right)^{-1/2}\sum_{n=0}^{\infty}\frac{\bar{z}^n}{\sqrt{x_n^{\gamma,\sigma}}!}\left|\psi_n\right\rangle
\end{equation}
where $\mathcal{N}_{\gamma,\sigma}(.)$ is the normalization factor and $\{\left|\psi_n\right\rangle\}$ is an orthonormal basis of an arbitrary (separable, infinite dimensional) Hilbert space $\mathcal{H}$. These $x_n^{\gamma,\sigma}$-NLCS are defined for all $z\in\mathbb{C}$, since $\lim\limits_{n\rightarrow +\infty}x_n^{\gamma,\sigma}=+\infty$.
Through equation \eqref{1.2},  the $x_n^{\gamma,\sigma}$-NLCSs \eqref{1.3} reads as
\begin{equation}
\label{1.4}
\left|z,\gamma,\sigma\right\rangle=\left(\mathcal{N}_{\gamma,\sigma}\left(z\bar{z}\right)\right)^{-1/2}\sum_{n=0}^{\infty}\sqrt{\frac{n!}{(2\gamma)_n}}\frac{\bar{z}^n}{(\sigma+1)_n}\left|\psi_n\right\rangle.
\end{equation} 

\begin{proposition}
Let $\gamma>0$ and $\sigma\in\mathbb{R}\backslash\mathbb{Z}^{\ast}_{-}$, be fixed parameters. Then, the normalization factor in \eqref{1.4} for the set of $x_n^{\gamma,\sigma}$-NLCSs reads
\begin{equation}
\label{1.6}
\mathcal{N}_{\gamma,\sigma}\left(z\bar{z}\right)={}_{2}F_{3}\left( 
\begin{array}{c}
1,1 \\ 
2\gamma,\sigma+1,\sigma+1%
\end{array}%
\Bigg\vert z\bar{z}\right)
\end{equation}
in terms of ${}_{2}F_{3}$-hypergeometric series. We have two interesting particular cases: 
\begin{itemize}
\item when $\gamma=1/2$ and $\sigma\in\mathbb{N}^{\ast}$, \eqref{1.6} reduces  to 
\begin{equation}
\label{3.7}
\mathcal{N}_{\frac{1}{2},\sigma}\left(z\bar{z}\right)=\left[I_{0}(2|z|)-\sum_{n=0}^{\sigma-1}\frac{|z|^{2n}}{(n!)^2}\right]\frac{\Gamma^2(\sigma+1)}{|z|^{2\sigma}}
\end{equation}
\item when $\sigma=0$, \eqref{1.6} reads 
\begin{equation}
\label{3.6}
\mathcal{N}_{\gamma,0}\left(z\bar{z}\right)=\Gamma(2\gamma)|z|^{1-2\gamma}I_{2\gamma-1}(2|z|),
\end{equation} 
for all $z\in\mathbb{C}$, where $I_{\tau}(.)$ is the bessel function of the first kind and of order $\tau$.
\end{itemize}
\end{proposition}

\begin{proof}  The inner product of two $x_n^{\gamma,\sigma}$-NLCSs expressed in \eqref{1.4} is given by 
\begin{eqnarray}
\label{1.7}
\left\langle z,\gamma,\sigma|w,\gamma,\sigma\right\rangle
&=& \left(\mathcal{N}_{\gamma,\sigma}\left(z\bar{z}\right)\mathcal{N}_{\gamma,\sigma}\left(w\overline{w}\right)\right)^{-1/2}\ _{2}F _{3}\left( 
\begin{array}{c}
1,1 \\ 
2\gamma,\sigma+1,\sigma+1%
\end{array}%
\Bigg\vert z\bar{w}\right).
\end{eqnarray}
If $\sigma+1$ is zero or a negative integer, the function $\ _{2}F _{3}$ is not defined, since all but a finite number of the terms of the series become infinite. Then, $\sigma\neq-1,-2,-3,...$ have to be satisfied. The normalization is deduced from \eqref{1.7} by taking $z = w$ such that $\left\langle z,\gamma,\sigma|z,\gamma,\sigma\right\rangle = 1$.
Now, we take $\gamma=1/2$ and $\sigma \in \mathbb{N}$ in \eqref{1.6} then 
\begin{eqnarray}
\label{3.111}
\mathcal{N}_{\frac{1}{2},\sigma}(z\bar{z})&=&{}_{1}F_{2}\left( 
\begin{array}{c}
1 \\ 
\sigma+1,\sigma+1%
\end{array}%
\Bigg\vert z\bar{z}\right)\\
&=&\frac{\Gamma^2(\sigma+1)}{|z|^{2\sigma}}\sum_{n\geq\sigma}\frac{|z|^{2n}}{\left(n!\right)^2}\\
&=& \left[I_{0}(2|z|)-\sum_{n=0}^{\sigma-1}\frac{|z|^{2n}}{(n!)^2}\right]\frac{\Gamma^2(\sigma+1)}{|z|^{2\sigma}}.
\end{eqnarray}
For $\sigma=0$, Eq.\eqref{1.6} becomes 
\begin{eqnarray*}
\label{1.8}
\mathcal{N}_{\gamma,0}(z\bar{z})=\ _{0}F _{1}\left( 
\begin{array}{c}
	- \\ 
	2\gamma%
\end{array}%
\Bigg\vert z\bar{z}\right)
\end{eqnarray*}
where the confluent hypergeometric function $_{0}F_{1}$ can  be written \cite[p.44]{srivastava}:
\begin{equation}
I_{\nu}(\xi)=\frac{\left(\frac{1}{2}\xi\right)^{\nu}}{\Gamma(\nu+1)}\ _{0}F _{1}\left( 
\begin{array}{c}
- \\ 
\nu+1%
\end{array}%
\Bigg\vert \frac{\xi^2}{4}\right).
\end{equation}
For parameters $\nu=2\gamma-1$ and $\xi=2|z|$, we obtain \eqref{3.6}.
\end{proof}
We see that the NLCSs \eqref{1.4} are over-complete, and do not form an orthonormal set. Two  vectors $|w,\gamma,\sigma\rangle$ and $|z,\gamma,\sigma\rangle$ are orthogonal if $\left\langle z,\gamma,\sigma|w,\gamma,\sigma\right\rangle = 0$  that means the entire function $\ _2F_3\left(1,1;2\gamma,\sigma+1,\sigma+1; z\bar{w}\right)$ has a zero at the point $z\bar{w}$. However, the entire function has no zeros if $z\bar{w}$ is a positive number.
\begin{remark}
The sequence $x_n^{\gamma,\sigma}$ can be identified in the litterature as a special case of sequences of positive numbers in the construction of the generalized hypergeometric CSs \cite{Appl, Popov}. However, the main reason for choosing the sequence $x_n^{\gamma,\sigma}$ resides in the fact that the corresponding NLCSs constructed in \eqref{1.3} generalize two specific states in the literature: 
\begin{itemize}
	\item Barut-Girardello coherent states \cite{Barut} while setting $2\gamma=1,2,3,...$ and $\sigma=0$;
	\item Philophase states \cite{Brif} while setting $\gamma=\frac{1}{2}$ and $\sigma\in \mathbb{N}^{\ast}$.
\end{itemize}
The generalized approach of such construction of NLCSs is the subject of a forthcoming paper.
\end{remark}
For all that follows, the general properties of $x_n^{\gamma,\sigma}$-NLCSs reduce to those of the above specific cases. 
\subsection{Resolution of identity}
Here, we will discuss the resolution of identity under some conditions on the parameters for which the $x_n^{\gamma,\sigma}$-NLCSs are well defined. The problem here is to find the measure $d\vartheta_{\gamma,\sigma}$ witch satisfy the following resolution of identity 
	\begin{equation}
	\label{3.11}
	\int_{\mathbb{C}}\left\vert z,\gamma,\sigma\right\rangle\ \left\langle z,\gamma,\sigma\right\vert d\vartheta
	_{\gamma,\sigma}(z)=1_{\mathcal{H}}.
	\end{equation}
To resolve this problem, we will write $d\vartheta_{\gamma,\sigma}$ as
\begin{equation}
\label{3.13}
d\vartheta _{\gamma,\sigma}(z)=\mathcal{N}_{\gamma,\sigma}\left(z\bar{z}\right)m_{\gamma,\sigma}\left(z\bar{z}\right)d\varrho(z),
\end{equation}
where $m_{\gamma,\sigma}$ is an auxiliary density function to be determined and $d\varrho $ is the Lebesgue measure on $\mathbb{C}$. By considering the polar coordinates $z=\rho e^{i\theta} ,\ \rho>0$ and $\theta\in [0,2\pi)$, the measure can be rewritten as 
\begin{equation}
d\vartheta _{\gamma,\sigma}(z)=\mathcal{N}_{\gamma,\sigma}\left(\rho ^{2}\right)m_{\gamma,\sigma}\left(\rho ^{2}\right)\frac{%
	\rho d\rho d\theta}{2\pi}.
\end{equation}
Using the expression \eqref{1.4} of $x_n^{\gamma,\sigma}$-NLCSs, the operator 
\begin{equation}
\mathcal{O}_{\gamma,\sigma}:=\int_{\mathbb{C}}\left| z,\gamma,\sigma\right\rangle
\left\langle z,\gamma,\sigma\right| d\vartheta_{\gamma,\sigma}(z)
\end{equation}
reads successively,
\begin{eqnarray}
\mathcal{O}_{\gamma,\sigma}&=&\sum\limits_{n,k=0}^{\infty} \sqrt{\frac{n!}{(2\gamma)_n}}\sqrt{\frac{k!}{(2\gamma)_k}}\frac{1}{(\sigma+1)_n(\sigma+1)_k}\nonumber\\
& & \times \left( \int_0^{\infty} 
\rho^{n+k}m_{\gamma,\sigma}\left(\rho^2\right)\left( \int_0^{2\pi} e^{i(k-n)\theta}\frac{d\theta}{2\pi} \right)\rho d\rho\right)
\vert \psi_n\rangle \langle \psi_k\vert \nonumber \\
&=&\sum\limits_{n=0}^{\infty}\frac{n!}{(2\gamma)_n}\frac{1}{(\sigma+1)^2_n}\left(
\int_0^{\infty}\rho^{2n}m_{\gamma,\sigma}(\rho^2)\rho d\rho \right) \vert \psi_n\rangle \langle
\psi_n\vert. 
\label{3.17}
\end{eqnarray}
We make the change of variable $r=\rho^2$, we obtain
\begin{equation}
	\mathcal{O}_{\gamma,\sigma} = \sum\limits_{n=0}^{\infty}\frac{n!}{2(2\gamma)_n}\frac{1}{(\sigma+1)^2_n}\left(
	\int_0^{\infty}r^{n}m_{\gamma,\sigma}(r)dr \right) \vert \psi_n\rangle \langle
	\psi_n\vert
\end{equation}
Thus, in order to recover the resolution of the identity \eqref{3.11}, we must have  
\begin{equation}  
\label{Meijer}
\int_0^{\infty}r^{n}m_{\gamma,\sigma}(r)dr=\frac{2(2\gamma)_n(\sigma+1)^2_n}{n!}.
\end{equation}
We start from the formula \cite[p.337]{Erdelyi1954}: 
\begin{equation}  \label{Meijerfirstintegral}
\int_{0}^{\infty}x^{s-1}G^{m l}_{p q} \left(x \  \Bigg\vert \  {a_1,\cdots,a_p\atop b_1,\cdots,b_q} \right)dx=\frac{%
	\prod\limits_{j=1}^{m}\Gamma(b_j+s)\prod\limits_{j=1}^{l}\Gamma(1-a_j-s)}{%
	\prod\limits_{j=l+1}^{p}\Gamma(a_j+s)\prod\limits_{j=m+1}^{q}\Gamma(1-b_j-s)}
\end{equation}
involving the Meijer's function $G_{pq}^{ml}$ with conditions $0\leq m\leq q$, $0\leq l\leq p$, $p+q<2(m+l)$, $-\min\limits_{1\leq j\leq m}\Re(b_j)<\Re(s)<1-\max\limits_{1\leq k\leq l}\Re(a_k)$. Then, for parameters $ x=r, \ m=3, \ l=0, \ p=1, \
q=3$, $a_1=0,\ b_1=2\gamma -1,\ b_2=b_3=\sigma $ and $%
s=n+1 $, equation \eqref{Meijerfirstintegral} reduces to 
\begin{equation}\label{Eq3.22}
\int_{0}^{\infty}r^{n}G^{3 0}_{1 3} \left(r \  \Bigg\vert \  {0\atop 2\gamma-1,\sigma,\sigma} \right)dr=\frac{\Gamma(n+2\gamma)\Gamma^2(n+\sigma+1)}{n!}
\end{equation}
that becomes
\begin{equation}  
\label{2.19}
\int_{0}^{\infty}\frac{2r^{n}}{\Gamma(2\gamma)\Gamma^2(\sigma+1)}G^{3 0}_{1 3} \left(r \  \Bigg\vert \  {0\atop 2\gamma-1,\sigma,\sigma} \right)dr=\frac{2(2\gamma)_n(\sigma+1)^2_n}{n!}.
\end{equation}
Comparing \eqref{2.19} to \eqref{Meijer} we obtain the searched weight function 
\begin{eqnarray}  
\label{weihtfunction1}
m_{\gamma,\sigma}(r)&=&\frac{2}{\Gamma(2\gamma)\Gamma^2(\sigma+1)}G^{3 0}_{1 3} \left(r \  \Bigg\vert \  {0\atop 2\gamma-1,\sigma,\sigma} \right),\ r>0.
\end{eqnarray}
We still need to prove the positiveness of $m_{\gamma,\sigma}(r)$, that remains to prove the positiveness of $G^{3 0}_{1 3}\left(r \big\vert{0\atop 2\gamma-1,\sigma,\sigma}\right)$. To accomplish this, we follow the work of Sixdeniers and Penson \cite{Sixdernier}. For this, we comeback to the equation \eqref{Eq3.22}  and we rewrite it by using the multiplication formula \cite[p.46]{srivastava}
\begin{equation}
y^{\nu}G^{m l}_{p q} \left(y \ \Bigg\vert \ {{(a_p)\atop(b_q)}}
\right)=G^{m l}_{p q} \left(y \ \Bigg\vert \ {{(a_p+\nu)\atop(b_q+\nu)}} \right),
\end{equation}
for $y=r$ and $\nu=1$ as
\begin{equation}
\int_{0}^{\infty}r^{n-1}G^{3 0}_{1 3} \left( r \  \Bigg\vert \  {1\atop 2\gamma,\sigma+1,\sigma+1} \right)dr=\frac{\Gamma(n+2\gamma)\Gamma^2(n+\sigma+1)}{n!}.
\end{equation}
This one is recognized as Mellin transform of $G^{3 0}_{1 3}\left(r\big\vert{1\atop 2\gamma,\sigma+1,\sigma+1} \right)$. Then, the inverse Mellin transform of $\frac{\Gamma(n+2\gamma)\Gamma^2(n+\sigma+1)}{n!}$ reads as 
\begin{equation}\label{Eq3.25}
G^{3 0}_{1 3} \left( r \  \Bigg\vert \  {1\atop 2\gamma,\sigma+1,\sigma+1} \right) = \frac{1}{2\pi i}\int_{-i\infty}^{+i\infty}\frac{\Gamma(n+2\gamma)\Gamma^2(n+\sigma+1)}{n!}r^{-n}dn.
\end{equation}
Now, according to the Mellin convolution property of inverse Mellin transform \cite{Marichev} also called the \textit{generalized Parseval formula}, if for arbitrary $f^{\ast}(n)$ and $g^{\ast}(n)$, there exists an inverse Mellin transform $f(r)$ and $g(r)$ respectively, then
\begin{equation}
\label{A.11}
\int_{0}^{\infty}f\left(\frac{r}{t}\right)g(t)\frac{1}{t}dt=\frac{1}{2\pi i}\int_{-i\infty}^{+i\infty}f^{\ast}(n)g^{\ast}(n)r^{-n}dn.
\end{equation}
We want to apply this property to equation \eqref{Eq3.25} in order to write the Meijer's function $G^{3 0}_{1 3}\left(r\big\vert{1\atop 2\gamma,\sigma+1,\sigma+1} \right)$ as an integrale representation of two positives functions $f(r)$ and $g(r)$. The choice of function $f^*(n)$ and $g^*(n)$ will yield restrictions on the parameters $\gamma$ and $\sigma$. There are four interesting cases that can be regrouped into the two following cases:
\begin{itemize}
\item Case 1:
\begin{equation}\label{case1}
f^{\ast}(n)=\frac{\Gamma(n+\textbf{a})\Gamma(n+\textbf{b})}{n!}, \qquad g^{\ast}(n)=\Gamma(n+\textbf{c})
\end{equation}
\item Case 2:
\begin{equation}\label{case2}
f^{\ast}(n)=\frac{\Gamma(n+\textbf{a})}{n!},\qquad	g^{\ast}(n)=\Gamma(n+\textbf{b})\Gamma(n+\textbf{c})
\end{equation}
\end{itemize} 
thus that the regrouped case 1, consists of the two cases $\textbf{a}=\textbf{b}=\sigma+1$, $\textbf{c}=2\gamma$ and $\textbf{a}=2\gamma$, $\textbf{b}=\textbf{c}=\sigma+1$ and the regrouped case 2 consists of the two cases $\textbf{a}=2\gamma$, $\textbf{b}=\textbf{c}=\sigma+1$ and $\textbf{a}=\textbf{c}=\sigma+1$, $\textbf{b}=2\gamma$.
\ \\ Note first that if one of parameters $\textbf{a,b}$ or $\textbf{c}$ is equal to $1$, then the meijer function $G^{3 0}_{1 3}\left(r\big\vert{1\atop \textbf{a},\textbf{b},\textbf{c}} \right)$ reduces to (for example $\textbf{a}=1$) \cite[p.61]{Mathai}:
\begin{equation}
G^{2 0}_{0 2}\left(r\big\vert{-\atop \textbf{b},\textbf{c}} \right) = 2r^{\frac{\textbf{b}+\textbf{c}}{2}} K_{\textbf{c}-\textbf{b}}(2\sqrt{r}),\ \ \ r>0
\end{equation}
that is positive (see \eqref{Mcd}). Otherwise, we have :
\ \\
\textbf{Case 1 discussion:} choosing $f^*(n)$ and $g^*(n)$ as given in \eqref{case1}
it will be identified respectively
\begin{equation}
f(r)=G^{2 0}_{1 2} \left( r \  \Bigg\vert \  {1\atop \textbf{a},\textbf{b}} \right), \quad \Re(\textbf{a},\textbf{b})>0
\end{equation}
and 
\begin{equation}
 g(r)=r^{\textbf{c}}e^{-r},\quad \Re(\textbf{c})>0.
\end{equation}
The function $g(r)$ is positive for $r\in[0,\infty)$. We have to prove the positiveness of $f(r)$. For this, we use the following lemma (see Appendix A for the proof):
\begin{lemma}
	An integral representation of the Meijer's function $G^{2 0}_{1 2}$ is given by
	\begin{equation}
	G^{2 0}_{1 2} \left( z \  \Bigg\vert \  {\alpha\atop \beta,\lambda} \right)=\frac{z^{\beta}e^{-z}}{\Gamma(\alpha-\lambda)} \int_{0}^{\infty}e^{-s z}s^{\alpha-\lambda-1}\left(1+s\right)^{\beta-\alpha}ds \label{IRM}
	\end{equation}
	for $|\arg z|<\frac{\pi}{2}$ and $\Re(\alpha-\lambda)>0$;
\end{lemma}
\noindent for parameters $\alpha=1$, $\beta=\textbf{a},\,\lambda=\textbf{b}$ and $z=\frac{r}{t}$, that require  the condition $\Re(\textbf{b})<1$, to get
\begin{equation}
f(r)=\frac{\left(\frac{r}{t}\right)^{\textbf{a}}e^{-\frac{r}{t}}}{\Gamma(1-\textbf{b})} \int_{0}^{\infty}e^{-\frac{sr}{t}}s^{-\textbf{b}}\left(1+s\right)^{\textbf{a}-1}ds.
\end{equation}
Clearly, $f(r)$ is positive for $\Re(\textbf{b})<1$.\\
Applying the property \eqref{A.11}  with $f^*(n)$, $g^*(n)$, $f(r)$ and $g(r)$ to the equation \eqref{Eq3.25}, we obtain the following integral representation
\begin{equation}\label{Eq3.32}
	G^{3 0}_{1 3} \left( r \  \Bigg\vert \  {1\atop \textbf{c},\textbf{a},\textbf{b}} \right) = \int_{0}^{\infty}g(t)f\left(\frac{r}{t}\right)\frac{1}{t}dt ,\quad \Re(\textbf{a,b,c})>0
\end{equation}
that assure well the positiveness of $G^{3 0}_{1 3}\left(r\big\vert{1\atop \textbf{c},\textbf{a},\textbf{b}} \right)$ for $\Re(\textbf{a,c})>0$ and $0<\Re(\textbf{b})<1$.
\begin{corollary}
	An integral representation of the Meijer's function $G^{3 0}_{1 3}$ is given by
	\begin{equation}
	G^{3 0}_{1 3} \left( r \  \Bigg\vert \  {1\atop a, b, c} \right)  = \int_0^{\infty}t^{a-1}e^{-t}G^{2 0}_{1 2} \left( \frac{r}{t}\  \Bigg\vert \  {1\atop b,c} \right)dt \label{Eq3.32}
	\end{equation}
	for $\Re(a,b,c)>0$. 
\end{corollary}
\noindent
\textbf{Case 2 discussion:} choosing $f^*(n)$ and $g^*(n)$ as given in \eqref{case1}, it will be identified by \cite[pp.195-196 ]{Fritz}:
\begin{equation}
f(r)=\left(\Gamma(1-\textbf{a})\right)^{-1}r^{\textbf{a}}(1-r)^{-\textbf{a}}\chi_{]0, 1[}(r),\quad 0<\Re(\textbf{a})<1
\end{equation}
and
\begin{equation}
\label{K0}
g(r)=2r^{\frac{1}{2}(\textbf{b}+\textbf{c})}K_{\textbf{b}-\textbf{c}}\left(2\sqrt{r}\right),\quad \Re(\textbf{b},\textbf{c})>0
\end{equation}
where $\chi_{]0, 1[}$ is the indicator function. Clearly $f(r)$ is positive and we can verify easily that $g(r)$ is positive by using the formula \eqref{Mcd}. Applying the property \eqref{A.11}  with $f^*(n)$, $g^*(n)$, $f(r)$ and $g(r)$ under the conditions $0<\Re(\textbf{a})<1$ and $\Re(\textbf{b},\textbf{c})>0$
to the equation \eqref{Eq3.25}, we obtain a second integral representation
\begin{equation}\label{Eq3.32}
G^{3 0}_{1 3} \left( r \  \Bigg\vert \  {1\atop \textbf{c},\textbf{a},\textbf{b}} \right) = \int_{0}^{\infty}g(t)f\left(\frac{r}{t}\right)\frac{1}{t}dt 
\end{equation}
that assure well the positiveness of $G^{3 0}_{1 3}\left(r\big\vert{1\atop \textbf{c},\textbf{a},\textbf{b}} \right)$ for $0<\Re(\textbf{a})<1$ and $\Re(\textbf{b},\textbf{c})>0$. 
\begin{corollary}
An integral representation of the Meijer's function $G^{3 0}_{1 3}$  is given by
\begin{equation}
G^{3 0}_{1 3} \left( r \  \Bigg\vert \  {1\atop a,b,c} \right)=\frac{2r^{\frac{a+c}{2}}}{\Gamma(1-b)}\int_1^{\infty}t^{\frac{a+c}{2}-1}(t-1)^{-b}K_{c-a}\left( 2\sqrt{rt}\right)dt,
\end{equation}
for $0<\Re(b)<1$ and $\Re(a,c)>0$.
\end{corollary}
With these two cases discussion, we have proved the positiveness of the (meijer function)  $G^{3 0}_{1 3}\left(r \big\vert{0\atop 2\gamma-1,\sigma,\sigma}\right)$ that remains some restrictions on the parameters $\gamma$ and $\sigma$, including conditions under which the NLCSs are normalized, that can be summarize in the following tabular: 
\begin{center}
\begin{tabular}{|c|c|c|c|}
	\hline
	Parameters & $f^{\ast}(n)$ & $g^{\ast}(n)$ & conditions on $\gamma$ and $\sigma$ \\ 
	\hline
	$\textbf{a}=\textbf{b}=\sigma +1$ & \multirow{2}{*}{$\frac{\Gamma(n+\sigma +1)\Gamma(n+\sigma +1)}{n!}$} & \multirow{2}{*}{$ \Gamma(n+2\gamma)$} & $0< \gamma $ \\
	$\textbf{c}=2\gamma$ & & &$ -1<\sigma \leq 0$ \\ \hline
	$\textbf{a}=2\gamma$ & \multirow{2}{*}{$\frac{\Gamma(n+2\gamma)\Gamma(n+\sigma +1)}{n!}$} & \multirow{2}{*}{$\Gamma(n+\sigma +1)$} & $0<\gamma $ \\
	$\textbf{b}=\textbf{c}=\sigma +1$ & & & $-1<\sigma \leq 0$ \\ \hline 
	$\textbf{a}=2\gamma$ & \multirow{2}{*}{$\frac{\Gamma(n+2\gamma)}{n!}$} & \multirow{2}{*}{$\Gamma(n+\sigma +1)\Gamma(n+\sigma +1)$} & $0<\gamma \leq 1/2$ \\
	$\textbf{b}=\textbf{c}=\sigma +1 $ & & & $-1<\sigma $ \\ \hline
	$\textbf{a}=\textbf{c}=\sigma +1$ & \multirow{2}{*}{$\frac{\Gamma(n+\sigma +1)}{n!}$} & \multirow{2}{*}{$\Gamma(n+2\gamma )\Gamma(n+\sigma +1)$} & $0<\gamma $ \\
	$\textbf{b}=2\gamma$ & & & $-1<\sigma \leq 0$ \\ \hline
\end{tabular}
\end{center}
\bigskip
Finally, we conclude that the weight function $m_{\gamma,\sigma}(z\bar{z})$ is positive for $(\gamma ,\sigma)\in S_1 \cup S_2$, where 
\begin{eqnarray}
S_1= ]0,+\infty [ \times ]-1,0],\quad \text{and} \quad S_2= ]0,1/2 ] \times ]-1,+\infty[.
\end{eqnarray}
 With this weight function, equation \eqref{3.17} reduces to $\mathcal{O}_{\gamma,\sigma}=\sum_{n=0}^{\infty}\left|\psi_n\right\rangle\left\langle \psi_n\right|=I_{\mathcal{H}}$ since $\left\{\left|\psi_n\right\rangle\right\}$ is an orthonormal basis of the Hilbert space $\mathcal{H}$. We obtain then the following result.

\begin{proposition}
	Let $(\gamma ,\sigma)\in S_1 \cup S_2$, be fixed parameters. Then, the set of NLCSs attached to the sequence $x_n^{\gamma,\sigma}$ defined in \eqref{1.4} satisfy the  resolution of identity 
	\begin{equation}
	\int_{\mathbb{C}}\left\vert z,\gamma,\sigma\right\rangle\ \left\langle z,\gamma,\sigma\right\vert d\vartheta
	_{\gamma,\sigma}(z)=I_{\mathcal{H}},
	\end{equation}
	in terms of a suitable measure given by
	\begin{equation}
	\label{resolutionidentity}
	d\vartheta _{\gamma,\sigma }(z)=\frac{2}{\Gamma(2\gamma)\Gamma^2(\sigma+1)}\ _{2}F _{3}\left( 
	\begin{array}{c}
	1,1 \\ 
	2\gamma,\sigma+1,\sigma+1
	\end{array}%
	\Bigg\vert z\bar{z}\right) G_{13}^{30}\left(z\bar{z}\ \Bigg\vert \ {{%
			0\atop 2\gamma-1,\ \sigma ,\sigma}}\right) d\varrho(z),
	\end{equation}%
where $G_{13}^{30}$ denotes the Meijer's G-function and 
	$d\varrho $ being the Lebesgue measure on $\mathbb{C}$. 
\end{proposition}

\begin{remark}
The photon-distribution of the NLCSs $\left|z,\gamma,\sigma\right\rangle$ is given by 
\begin{equation}
P_n(z,\gamma,\sigma)=\left[\ {}_{2}\digamma_{3}\left( 
\begin{array}{c}
1,1 \\ 
2\gamma,\sigma+1,\sigma+1%
\end{array}%
\Bigg\vert z\bar{z}\right)\right]^{-1}\frac{n!}{(2\gamma)_n}\frac{|z|^{2n}}{(\sigma+1)_n^2}.
\end{equation}
The sequences of positive number $E_n:=x_n^{\gamma,\sigma}$ define in \eqref{3.1} is the spectrum of an unknown Hamiltonian.
\end{remark}
\section{$x_n^{\gamma,\sigma}$-NLCSs attached to the pseudoharmonic oscillator}
In this section, we discuss the closed form of the constructed NLCSs given in \eqref{1.3} by choosing the orthonormal basis $\left|\psi_n\right\rangle$ of the Hilbert space $\mathcal{H}$ as the eigenfunctions of the pseudoharmonic oscillator $\Delta_{\alpha,\beta}$. Then, we define the associated Bargmann-type transform and we derive some interesting formulas.
\subsection{The pseudoharmonic oscillator $\Delta_{\alpha,\beta}$}\ \\
An anharmonic potential that can be used to calculate the vibrational energies of a diatomic molecule has the form
\begin{equation}
V_{\varrho,\kappa_{0}}(\xi)=\varrho\left(\frac{\xi}{\kappa_{0}}-\frac{\kappa_{0}}{\xi}
\right)^{2}
\end{equation}
where $\kappa_{0}>0$ denotes the equilibrium bond length which is the distance between the diatomic nuclei, and $\varrho>0$ with $\varrho\kappa_0^{-2}$ represents a constant force. The associated stationary Schr\"{o}dinger equation reads
\begin{equation}
\label{s31}
-\frac{d^{2}}{d\xi^{2}}\psi(\xi)+\varrho\left(\frac{\xi}{\kappa_{0}}-\frac{\kappa_{0}}{\xi}
\right)^{2}\psi(\xi)=\lambda\psi(\xi),
\end{equation}
where $\psi$ satisfies the Dirichlet boundary condition $\psi(0)=0$. It is an exactly solvable equation.
To simplify the notation we introduce the new parameters $\alpha:=\varrho\kappa_{0}^{2}$ and $\beta:=\kappa_{0}^{-1}\sqrt{\varrho}$. Thereby the Hamiltonian in \eqref{s31} takes the form
\begin{equation}
\label{s2}
\Delta_{\alpha,\beta}:=-\frac{d^{2}}{d\xi^{2}}+\beta^{2}\xi^{2}+\frac{\alpha}{\xi^{2}},\ \ \xi\in\mathbb{R}_{+},\ \ \beta,\alpha>0
\end{equation}
called Gol'dman-Krivchenkov Hamiltonian \cite{Nassar}. Its spectrum in the Hilbert space $L^{2}(\mathbb{R}_{+},d\xi)$ reduces to a discrete part consisting of eigenvalues of the form \cite[pp.9-10]{Saad2}:
\begin{equation}
\label{s1}
\lambda_{n}^{\mu,\beta}:=2\beta(2n+\mu),\ \ \mu=\mu (\alpha)=1+\frac{1}{2}\sqrt{1+4\alpha}> \frac{3}{2},\ \ n=0,1,2,...
\end{equation}
and the wavefunctions of the corresponding normalised eigenfunctions are given by
\begin{equation}
\label{ss}
\langle\xi|\psi_{n}^{\mu,\beta}\rangle:=\left(\frac{2\beta^{\mu}n!}{\Gamma(\mu+n)}\right)^{\frac{1}{2}}\xi^{\mu-\frac{1}{2}}e^{-\frac{1}{2}\beta\xi^{2}}L_{n}^{(\mu-1)}(\beta\xi^{2}),\ \ n=0,1,2,...
\end{equation}
in terms of the Laguerre polynomials $L_{n}^{(\alpha)}$. The set of the functions \eqref{ss} constitutes a complete orthonormal basis for the Hilbert space $L^{2}(\mathbb{R}_{+},d\xi)$.

\begin{remark}
	We should note that the eigenvalues together with the eigenfunctions of $\Delta_{\alpha,\beta}$ could also be obtained by using raising and lowering operators through a factorization of this Hamiltonian based on the Lie algebra $SU(1,1)$ commutation relations \cite[pp.3-4]{Nakiev}.
\end{remark}

\subsection{The $x_n^{\gamma,\sigma}$-NLCS attached to $\Delta_{\alpha,\beta}$}\ \\
As announced in section 4, choosing the Hilbert space  $\mathcal{H}=L^2\left(\mathbb{R}_{+},d\xi\right)$ with its orthonormal basis $\left|\psi_n\right\rangle :=\left|\psi^{\mu,\beta}_n\right\rangle$ we discuss in the following proposition the closed form of the $x_n^{\gamma,\sigma}$-NLCS 
\begin{equation}
\label{4.1}
\left|z, \gamma,\sigma\right\rangle_{\mu,\beta}=
\left(\mathcal{N}_{\gamma,\sigma}\left(z\bar{z}\right)\right)^{-1/2}\sum_{n=0}^{\infty}\sqrt{\frac{n!}{(2\gamma)_n}}\frac{\bar{z}^{n}}{(\sigma+1)_n}\left|\psi^{\mu,\beta}_n\right\rangle,\quad z\in\mathbb{C}
\end{equation}
where $\gamma>0$ and $\sigma\leq0; \ \sigma \neq -1, -2, -3,...$ and $\mathcal{N}_{\gamma,\sigma}$ is the normalization factor in \eqref{1.6}.
\begin{proposition}\label{prop4.2.1}
	Let $2\gamma=\mu >\frac{3}{2}$, $\beta>0$ and $\sigma\leq 0$, be fixed parameters. Then, the wavefunctions of the states $\left|z, 2^{-1}\mu,\sigma\right\rangle_{\mu,\beta}$ defined in \eqref{4.1} can be written as 
	\begin{equation}
	\label{5.2}
	\begin{split}
	\left\langle \xi|z, 2^{-1}\mu,\sigma\right\rangle_{\mu,\beta} =\, & \sqrt{\frac{2\beta^{\mu}}{\Gamma(\mu)}}\left[\ _{2}F _{3}\left( 
	\begin{array}{c}
	1,1 \\ 
	\mu,\sigma+1,\sigma+1%
	\end{array}%
	\Bigg\vert z\bar{z}\right)\right]^{-\frac{1}{2}} \\
	& \times  \xi^{\mu-\frac{1}{2}}e^{-\frac{1}{2}\beta\xi^2}F^{{1:0;0}}_{1:0;1}\left( 
	\begin{array}{c}
	1:-;-; \\ 
	\sigma+1;-;\mu;%
	\end{array}%
	\Bigg\vert \bar{z},-\beta \bar{z}\xi^2\right)
	\end{split}
	\end{equation}
	where $F^{{p:q;k}}_{l:m;n}$ is the generalized Kamp\'{e} de F\'{e}riet's hypergeometric function \cite[p.63]{srivastava}.\\
	In particular, for $\sigma=0$, equation \eqref{5.2} reduces to 
	\begin{equation}
	\label{5.3}
	\left\langle \xi|z, 2^{-1}\mu,0\right\rangle_{\mu,\beta}=\sqrt{2\beta}\left(\left(\bar{z}|z|^{-1}\right)^{\mu-1}I_{\mu-1}(2|z|)\right)^{-\frac{1}{2}}\sqrt{\xi}e^{-\frac{1}{2}\beta\xi^2+\bar{z}}J_{\mu-1}\left(2\xi\sqrt{\beta\bar{z}}\right),\,\, z\in\mathbb{C}
	\end{equation}
 where $I_{\mu}(.)$ and $J_{\mu}(.)$ are the modified Bessel functions of the first kind of order $\mu$.
\end{proposition}

\begin{proof}[\textbf{Proof.}]
We start by writing the expression of the wave function of states $\left|z;\gamma,\mu\right\rangle_{\mu,\beta}$ according to
definition \eqref{4.1} as
\begin{equation}
\left\langle \xi|z, \gamma,\sigma\right\rangle_{\mu,\beta}=\left(\mathcal{N}_{\gamma,\sigma}\left(z\bar{z}\right)\right)^{-1/2}\sum_{n=0}^{\infty}\sqrt{\frac{n!}{(2\gamma)_n}}\frac{\bar{z}^n}{(\sigma+1)_n} \left\langle \xi|\psi^{\mu,\beta}_n\right\rangle,\ \xi\in\mathbb{R}_{+}.
\end{equation}
Now, putting $2\gamma=\mu$ and using the expression of $\left\langle \xi|\psi^{\mu,\beta}_n\right\rangle$ given in \eqref{ss}, we obtain 
\begin{equation}
\label{4.4}
\left\langle \xi|z, 2^{-1}\mu,\sigma\right\rangle_{\mu,\beta}=\left(\mathcal{N}_{\frac{\mu}{2},\sigma}\left(z\bar{z}\right)\right)^{-1/2}\sqrt{\frac{2\beta^{\mu}}{\Gamma(\mu)}}\xi^{\mu-\frac{1}{2}}e^{-\frac{1}{2}\beta\xi^2} \sum_{n=0}^{\infty}\frac{n!}{\left(\mu\right)_n}\frac{\bar{z}^n}{\left(\sigma+1\right)_n} L_n^{(\mu-1)}\left(\beta\xi^2\right).
\end{equation}
In order to compute the infinite sum in the right hand side of the above equation, we use the relation between the Laguerre polynomials and the Kummer's function \cite[p.240]{Magnus}:
\begin{equation}
_{1}F _{1}\left( 
\begin{array}{c}
-k \\ 
b%
\end{array}%
\Bigg\vert x\right)=\frac{k!}{(b)_k}L_k^{(b-1)}(x).
\end{equation}
Then, with $b=\mu$
\begin{eqnarray}
\label{4.5}
\mathfrak{S}_{\sigma,\mu}(\xi^2)&:=&\sum_{n=0}^{\infty}\frac{n!}{\left(\mu\right)_n}\frac{\bar{z}^n}{\left(\sigma+1\right)_n} L_n^{(\mu-1)}\left(\beta\xi^2\right) \\
	& = & \sum_{n=0}^{\infty}\frac{\left(\bar{z}\right)^n}{\left(\sigma+1\right)_n}\ _{1}F_{1}\left( 
	\begin{array}{c}
	-n \\ 
	\mu%
	\end{array}%
	\Bigg\vert \beta\xi^2\right). \label{Ssu}
\end{eqnarray}
Next, with the help of the generating formula \cite[p.165]{srivastava}:
\begin{equation}
\label{4.8}
\sum_{n=0}^{\infty}\frac{\prod\limits_{j=1}^{r}(c_j)_n}{\prod\limits_{j=1}^{s}(d_j)_n} \ _{p+1}F _{q}\left( 
\begin{array}{c}
-n; (a_p) \\ 
(b_q)%
\end{array}%
\Bigg\vert x\right)\frac{t^n}{n!}=F^{{r:0;p}}_{s:0;q}\left( 
\begin{array}{c}
(c_r):-;(a_p); \\ 
(d_s):-;(b_q);%
\end{array}%
\Bigg\vert t,-xt\right),
\end{equation}
where $F^{{r:0;p}}_{s:0;q}$ is the generalized Kamp\'{e} de F\'{e}riet's hypergeometric function of two variables (see \cite[p.63]{srivastava}), we obtain through the parameters matching $r=s=q=1$, $p=0$, $c_1=1$, $b_1=\mu$, $d_1=\sigma+1$, $t=\bar{z}$ and $x=\beta\xi^2$:  
\begin{equation}
\label{4.11}
\mathfrak{S}_{\sigma,\mu}(\xi^2)=F^{{1:0;0}}_{1:0;1}\left( 
\begin{array}{c}
1:-;-; \\ 
\sigma+1;-;\mu;%
\end{array}%
\Bigg\vert \bar{z},-\beta \bar{z}\xi^2\right),
\end{equation}
which converge for all $|\bar{z}|<\infty$ and $\beta|\bar{z}\xi^2|<\infty$. 
Finally, go back to \eqref{4.4} with \eqref{4.11} and replace the expression of $\mathcal{N}_{\frac{\mu}{2},\sigma}\left(z\bar{z}\right)$ given in \eqref{1.6}, to obtain the announced result in \eqref{5.2}. \\
Now, for the case $\sigma=0$, we have
\begin{eqnarray}
\mathfrak{S}_{0,\mu}(\xi^2) = F^{{0:0;0}}_{0:0;1}\left( 
\begin{array}{c}
\text{-}:\text{-};\text{-}; \\ 
\text{-};\text{-};\mu;
\end{array}
\Bigg\vert \bar{z},-\beta\bar{z}\xi^2\right) &=& \sum_{r,s\geq 0} \frac{1}{(\mu)_s} \frac{\bar{z}^r}{r!} \frac{(-\beta\bar{z}\xi^2)^s}{s!} \nonumber \\
&=& e^{\bar{z}} \sum_{s\geq 0} \frac{1}{(\mu)_s} \frac{(-\beta\bar{z}\xi^2)^s}{s!} \nonumber\\
&=& e^{\bar{z}} \Gamma(\mu) (\xi\sqrt{\beta\bar{z}})^{1-\mu} J_{\mu-1}\left(2\xi\sqrt{\beta\bar{z}}\right) \label{O0u}
\end{eqnarray}
where $J_{\nu}$ is the Bessel function \cite[p.65]{Magnus}. Proceeding like in the general case, we get the announced result in \eqref{5.3}. This ends the proof.
\end{proof}

\subsection{A Bargmann-type transform}
Next, once we have obtained the closed form, we can associate to the $x_n^{\gamma,\sigma}$-NLCSs in \eqref{4.1} the Bargmann-type transform which will make a connection between the Hilbert space $L^2\left(\mathbb{R}_{+},d\xi\right)$ of the physical system and the space of coefficients. For this, let us recall that the reproducing kernel arising from this $x_n^{\gamma,\sigma}$-NLCS with the connection parameters $\gamma=\mu/2$, reads 
\begin{equation}
\mathcal{K}(z,\overline{w})=\ _{2}F_{3}\left( 
\begin{array}{c}
1,1 \\ 
\mu,\sigma+1,\sigma+1%
\end{array}%
\Bigg\vert z\bar{w}\right).
\end{equation}
The corresponding reproducing kernel Hilbert space consisting of functions which are holomorphic in $\mathbb{C}$, denoted by $\mathcal{A}_{\mu,\sigma}(\mathbb{C})$, is a subspace of the larger Hilbert space $L^{2}(\mathbb{C},d\nu_{\mu,\sigma})$ where the measure $d\nu_{\mu,\sigma}\left(z,\bar{z}\right)$ is given by
\begin{equation}
d\nu_{\mu,\sigma}\left(z,\bar{z}\right)=\frac{2}{\Gamma(\mu)\Gamma^{2}\left(\sigma+1\right)}G^{3 0}_{1 3} \left(z\bar{z} \  \Bigg\vert \  {0\atop \mu-1,\sigma,\sigma} \right)d\varrho(z).
\end{equation} 
In view of the resolution of the identity, we see that the map $\mathcal{B}_{\mu,\sigma}:L^{2}(\mathbb{R}_{+},d\xi)\longrightarrow \mathcal{A}_{\mu,\sigma}(\mathbb{C})$ defined by 
\begin{equation}
\label{4.25}
\mathcal{B}_{\mu,\sigma}[\varphi]\left(z\right)=\left(\mathcal{N}_{\frac{\mu}{2},\sigma}\left(z\bar{z}\right)\right)^{1/2}\langle \varphi|z, 2^{-1}\mu,\sigma\rangle_{\mu,\beta}
\end{equation}
is a unitary map, embeding $L^{2}(\mathbb{R}_{+},d\xi)$ into the holomorphic subspace $\mathcal{A}_{\mu,\sigma}(\mathbb{C})\subset L^{2}(\mathbb{C},d\nu_{\mu,\sigma})$. In order to express it as an integral transform we make use of Proposition \ref{prop4.2.1}.
\begin{theorem}
	Let $\mu>3/2$, $-1<\sigma\leq0$ and $\beta>0$, be fixed parameters. Then, the Bargmann-type transform is the unitary map $\mathcal{B}_{\mu,\sigma}:L^{2}(\mathbb{R}_{+},d\xi)\longrightarrow \mathcal{A}_{\mu,\sigma}(\mathbb{C})$ defined by means of \eqref{5.2} as
	\begin{equation}
	\label{5.17}
	\mathcal{B}_{\mu,\sigma}[\varphi]\left(z\right)=\sqrt{\frac{2\beta^{\mu}}{\Gamma(\mu)}}\int_{0}^{\infty}\xi^{\mu-\frac{1}{2}}e^{-\frac{1}{2}\beta\xi^{2}}F^{{1:0;0}}_{1:0;1}\left( 
	\begin{array}{c}
	1:-;-; \\ 
	\sigma+1;-;\mu;%
	\end{array}%
	\Bigg\vert z,-\beta z\xi^2\eta\right)\varphi(\xi)d\xi
	\end{equation}
	for every $z\in \mathbb{C}$.
	For $\sigma=0$, it reduces to
	\begin{equation}
	\label{5.18}
	\mathcal{B}_{\mu,0}[\varphi](z)=\sqrt{2\beta\Gamma(\mu)}z^{\frac{1-\mu}{2}}e^{z}\int_{0}^{\infty}\xi^{\frac{1}{2}}e^{-\frac{1}{2}\beta\xi^2} J_{\mu-1}\left(2\sqrt{\beta z}\xi\right)\varphi(\xi)d\xi.
	\end{equation}
	
\end{theorem} 
With the help of the resolution of the identity \eqref{3.11} and the transform \eqref{4.25}, we can represent any arbitrary state $|\varphi\rangle$ in $L^{2}(\mathbb{R}_{+},d\xi)$ in terms of the NLCSs \eqref{4.1} as follows:
\begin{equation}
|\varphi\rangle=\int_{\mathbb{C}}d\nu_{\mu,\sigma}(z)\left(\mathcal{N}_{\frac{\mu}{2},\sigma}\left(z\bar{z}\right)\right)^{-\frac{1}{2}}\mathcal{B}_{\mu,\sigma}[\varphi](z)|z, 2^{-1}\mu,\sigma\rangle_{\mu,\beta}.
\end{equation}
Therefore, the norm square of $|\varphi\rangle$ also reads
\begin{equation}
\langle\varphi|\varphi\rangle_{\mathcal{H}}=\frac{2}{\Gamma(\mu)\Gamma^{2}(\sigma+1)}\int_{\mathbb{C}}|\mathcal{B}_{\mu,\sigma}[\phi](z)|^{2}G_{13}^{30}\left(|z|^2\ \Bigg\vert \ {{%
		0\atop \mu-1,\ \sigma,\sigma}}\right)d\varrho(z).
\end{equation}
The carefull reader has certaintly recognized in \eqref{5.18} the Bessel transform \cite{Fritz}. As a consequence we have the following result.
\begin{corollary}
	Let $\mu>3/2$,  $-1<\sigma\leq0$ and $\beta>0$, be fixed parameters. The following integral 
	\begin{equation}
	\label{5.19}
	\int_{0}^{\infty}\xi^{2\mu-1}e^{-\beta\xi^{2}}F^{{1:0;0}}_{1:0;1}\left( 
	\begin{array}{c}
	1:-;-; \\ 
	\sigma+1;-;\mu;%
	\end{array}%
	\Bigg\vert z,-\beta z\xi^2\right)L_n^{(\mu-1)}(\beta\xi^2)d\xi=\frac{\Gamma(\mu)}{2\beta^{\mu}} \frac{z^n}{\left(\sigma+1\right)_n}
	\end{equation}
	holds true for every $z\in\mathbb{C}$. When $\sigma=0$, it reduces to 
	\begin{equation}
	\label{5.20}
	\int_{0}^{\infty}\xi^{\mu}e^{-\beta\xi^2} J_{\mu-1}\left(2\sqrt{\beta z}\xi\right)L_n^{(\mu-1)}(\beta\xi^2)d\xi=\frac{z^{n+\frac{\mu-1}{2}}}{n!}\frac{e^{-z}}{2\beta^{\frac{\mu+1}{2}}}.
	\end{equation}
	
\end{corollary}
\begin{proof}
According to \eqref{4.25}, the range of the basis vector $\left\{\psi_n^{\mu,\beta}\right\}$ under the transform $\mathcal{B}_{\mu,\sigma}$ should exactly be the coefficients in \eqref{4.1} with the parameter $2\gamma=\mu$. More precisely  
\begin{equation}
\label{5.21}
\mathcal{B}_{\mu,\sigma}[\psi_n^{\mu,\beta}](z)=\sqrt{\frac{n!}{(\mu)_n}}\frac{z^n}{\left(\sigma+1\right)_n}.
\end{equation}
On the other hand, applying the transform \eqref{5.17} to $\psi_n^{\mu,\beta}$ one obtains
\begin{equation}
\label{BPsi2}
\mathcal{B}_{\mu,\sigma}[\psi_n^{\mu,\beta}](z) = \frac{2\beta^{\mu}}{\Gamma(\mu)}\sqrt{\frac{n!}{(\mu)_n}}
\int_{0}^{\infty}\xi^{2\mu-1}e^{-\beta\xi^{2}}F^{{1:0;0}}_{1:0;1}\left( 
\begin{array}{c}
1:-;-; \\ 
\sigma+1;-;\mu;%
\end{array}%
\Bigg\vert z,-\beta z\xi^2\right)L_n^{(\mu-1)}(\beta\xi^2)d\xi.
\end{equation}
Then, matching \eqref{5.21} and \eqref{BPsi2} results to the announced formula \eqref{5.19}. Now, taking $\sigma=0$ in \eqref{5.19} it follows the formula \eqref{5.20} using \eqref{O0u}. This ends the proof.
\end{proof}
\begin{remark}
Note that \eqref{5.20} can be obtained by using the formula \cite[p.65]{Fritz}:
\begin{equation}
\begin{split}
\int_{0}^{\infty}x^{\frac{1}{2}+\nu}\exp\left(-a^2x^2\right)L^{(\nu)}_n(a^2x^2)(xy)^{1/2}J_{\nu}(xy)dx = \, \frac{y^{2n+\nu+1/2}}{n!2^{2n+\nu+1}} \frac{\exp\left(-\frac{y^2}{4a^2}\right)}{a^{2(\nu +n+1)}},
\end{split}
\end{equation}
with $\nu>-1$, for parameters $x=\xi$, $y=2\sqrt{\beta z}$, $a=\sqrt{\beta}$ and $\nu=\mu-1$. 
\end{remark}

\section{Conclusion}
 In the present paper, we have constructed a new family of nonlinear coherent states (NLCSs) by replacing the factorial $n!$ occurring in the coefficients of the canonical CSs by a specific generalized factorial $x_{n}^{\gamma,\sigma }!=x_{1}^{\gamma,\sigma }\cdots x_{n}^{\gamma,\sigma }$ with $x_{0}^{\gamma,\sigma }=0$, where $x_{n}^{\gamma,\sigma }$ is a sequence of positive numbers. We have showed, for $(\gamma,\sigma)\in \left(]0,+\infty [ \times ]-1,0]\right) \cup \left(]0,1/2 ] \times ]-1,+\infty ]\right)$, that these states satisfy a resolution of identity relation in the Fock space and that the obtained NLCSs includes as special cases both the Barut-Giraredello CSs and Philophase states. The new coefficients are then used to consider a superposition of eigenstates of the pseudoharmonic oscillator $\Delta_{\alpha ,\beta}$. The obtained states constitute a four-parameters family of NLCSs under the same restrictions on the parameters $\gamma$ and $\sigma$. For equal parameters $\gamma =2^{-1}\mu(\alpha) $, we define the associated Bargmann-type transform and we derive some integral formulas. 
\begin{appendix}
\section{Proof of the lemma 1}
We use the formula \cite[p.62]{Mathai}
\begin{equation}
G^{2 0}_{1 2} \left( z \  \Bigg\vert \  {\alpha\atop \beta,\lambda} \right)=z^{\frac{\beta+\lambda-1}{2}}e^{-\frac{z}{2}}W_{\frac{1+\beta+\lambda}{2}-\alpha,\frac{\beta-\lambda}{2}}(z)
\end{equation}
where $W_{\nu,\mu}$ is the Whittaker function and the integral representation of $W_{\nu,\mu}$
\cite[p.1025]{Gradstyn}
\begin{equation}
W_{\nu,\mu}(z)=\frac{z^{\mu+\frac{1}{2}}e^{-\frac{z}{2}}}{\Gamma(\mu-\nu+\frac{1}{2})}\int_{0}^{\infty}e^{-s z}s^{\mu-\nu-\frac{1}{2}}\left(1+s\right)^{\mu+\nu-\frac{1}{2}}ds
\end{equation}
under the conditions $|\arg z|<\frac{\pi}{2}$ and $\Re(\mu-\nu)>-\frac{1}{2}$ for $\nu=\frac{1+\beta +\lambda}{2}-\alpha$ and $\mu=\frac{\beta-\lambda}{2}$ to obtain the integral representation of the Meijer's function $G^{2 0}_{1 2}$ 
\begin{equation}
G^{2 0}_{1 2} \left( z \  \Bigg\vert \  {\alpha\atop \beta,\lambda} \right)=\frac{z^{\beta}e^{-z}}{\Gamma(\alpha-\lambda)} \int_{0}^{\infty}e^{-s z}s^{\alpha-\lambda-1}\left(1+s\right)^{\beta-\alpha}ds,
\end{equation}
for $|\arg z|<\frac{\pi}{2}$ and $\Re(\alpha-\lambda)>0$.

	\end{appendix}

\end{document}